\documentclass[11pt]{article}
\usepackage{latexsym}
\usepackage{amsmath}
\usepackage{amssymb}
\usepackage{mathrsfs}
\usepackage{amsthm}

\usepackage{graphics}
\newtheorem{theorem}{Theorem}
\newtheorem{lemma}{Lemma}

\newtheorem{claim}{Claim}

\usepackage{times}
\usepackage{vmargin}
\usepackage{mathrsfs}

\setmarginsrb{1in}{1in}{1in}{1in}{0cm}{0.0cm}{0cm}{0.6cm}

\newcommand{\IOB}{\textsc{$k$-IOB}}
\newcommand{\rIOB}{\textsc{R-$k$-IOB}}
\newcommand{\LOB}{\textsc{$k$-LOB}}
\newcommand{\rLOB}{\textsc{R-$k$-LOB}}
\newcommand{\ve}{\varepsilon}

\newcommand{\cO}{\mathcal{O}}

\DeclareMathOperator{\operatorClassNP}{NP}
\newcommand{\classNP}{\ensuremath{\operatorClassNP}}

\newcommand{\tw}{{\mathbf{tw}}}

\title{Beyond  Bidimensionality: Parameterized Subexponential  Algorithms on Directed Graphs}

\author{Frederic Dorn\thanks{Department of Informatics,
 University of Bergen, Norway.\newline 
$~$\hspace{.5cm}{\tt {\{dorn|fedor.fomin|daniello|saket.saurabh\}@ii.uib.no}}.}\addtocounter{footnote}{-1}
 \and Fedor V. Fomin\footnotemark 
\and \addtocounter{footnote}{-1}Daniel Lokshtanov \footnotemark 
\and Venkatesh Raman\thanks{The Institute of Mathematical Sciences, Chennai, India. \newline
$~$\hspace{.5cm}{\tt vraman@imsc.res.in}}\addtocounter{footnote}{-2} 
\and  Saket Saurabh\footnotemark
}

\begin{document}

\date{}
\maketitle

\begin{abstract}
In 2000 Alber et al.~[{\em SWAT 2000}\,] obtained the first parameterized subexponential algorithm on undirected planar graphs by showing that $k$-{\sc Dominating Set} is solvable in time $2^{\cO(\sqrt{k})} n^{\cO(1)}$, where $n$ is the input size. This result triggered an extensive study of parameterized problems on planar and more general classes of sparse graphs and culminated in  the creation of 
 Bidimensionality Theory by Demaine et al.~[{\em J. ACM 2005}\,]. The theory utilizes deep theorems from Graph Minor Theory of Robertson and Seymour, and provides a simple criteria for checking whether a parameterized problem is solvable in subexponential time on sparse graphs. 

While bidimensionality theory is an algorithmic framework  on undirected graphs, it remains unclear how to apply it to problems on directed graphs.  The main reason is that Graph Minor Theory for directed graphs is still in a nascent stage and there are no suitable obstruction theorems so far.  Even the analogue  of treewidth for directed graphs is not unique and several alternative definitions have been proposed. 

In this paper we make the first step beyond bidimensionality by obtaining subexponential time algorithms for problems on directed graphs. 
We 
develop two different methods to achieve subexponential time parameterized algorithms for problems on sparse directed graphs. 
 We exemplify our approaches with two well studied problems.
  For the first problem,   {\sc $k$-Leaf Out-Branching}, which is to find an oriented spanning tree with at least $k$ leaves,  we obtain an algorithm solving the problem in time
   $2^{\cO(\sqrt{k} \log k)} n+ n^{\cO(1)}$ on directed graphs whose underlying  undirected graph excludes some fixed graph $H$ as a minor.  For the special case when the input directed graph is planar, the running time can be improved to  $2^{\cO(\sqrt{k} )}n + n^{\cO(1)}$.
   The  second example is 
a generalization of the {\sc Directed Hamiltonian Path} problem, namely 
   {\sc $k$-Internal Out-Branching}, which is to find an oriented  spanning tree with at least  $k$ internal vertices. We obtain an algorithm solving the problem in time $2^{\cO(\sqrt{k} \log k)} + n^{\cO(1)}$ on directed graphs whose underlying  undirected graph excludes some fixed apex graph $H$ as a minor. 
   Finally, we observe that   for any $\ve>0$, the  {\sc $k$-Directed Path} problem is solvable in time 
  $\cO((1+\ve)^k n^{f(\ve)})$, where  $f$ is some function of $\ve$.
   
Our methods are based on non-trivial combinations of obstruction theorems for undirected graphs, kernelization, problem specific combinatorial structures 
and a layering technique similar to the one employed by Baker to obtain PTAS for planar graphs.   
\end{abstract}
\section{Introduction}
\label{intro}

Parameterized complexity theory is a framework for a refined analysis of hard (\classNP{}-hard) problems.  
Here, 
every input instance $I$ of a problem $\Pi$ is accompanied with an integer parameter $k$ and $\Pi$ is said to be fixed parameter tractable (FPT) if there is an algorithm 
running in time $f(k)\cdot n^{\cO(1)}$, where $n=|I|$ and $f$ is a computable function. A central problem in parameterized algorithms is to obtain algorithms with running time 
$f(k)\cdot n^{\cO(1)}$ such that $f$ is as slow growing function as possible. This has led to the development of various graph algorithms with running time $2^{\cO(k)}n^{O(1)}$--- 
notable ones include {\sc $k$-Feedback Vertex Set}~\cite{ChenFLLV08}, {\sc $k$-Leaf Spanning Tree}~\cite{KLR08}, {\sc $k$-Odd Cycle Transversal}~\cite{ReedSV04}, {\sc $k$-Path}~\cite{AlonYZ95}, and  
{\sc $k$-Vertex Cover}~\cite{ChenKX06} in undirected graphs. A natural question was whether we can get {\em subexponential time} algorithms for these problems, that is, can we have algorithms with 
running time $2^{o(k)}n^{\cO(1)}$. It is now possible to show that these problems do not admit algorithms with running time $2^{o(k)}n^{\cO(1)}$ unless the  
 Exponential Time Hypothesis (ETH)~\cite{FlumGrohebook,ImpagliazzoPZ01} fails.  
%
Finding algorithms with subexponential running time on general undirected graphs is a trait uncommon to parameterized algorithms.

However, the situation changes completely when we consider problems on topological graph classes like planar graphs or graphs of bounded genus.  
In $2000$, Alber et al.~\cite{AlberBFKN02} obtained the first parameterized subexponential algorithm on undirected planar graphs by 
showing that $k$-{\sc Dominating Set} is solvable in time $2^{\cO(\sqrt{k})} n^{\cO(1)}$. This result triggered an extensive study of parameterized problems on planar and more general classes of 
sparse graphs like graphs of bounded genus, apex minor-free graphs and $H$-minor free graphs. All this work led to subexponential time algorithms for several fundamental 
problems like {\sc $k$-Feedback Vertex Set}, {\sc $k$-Edge Dominating Set}, {\sc $k$-Leaf Spanning Tree}, {\sc $k$-Path}, $k$-{\sc $r$-Dominating Set}, {\sc $k$-Vertex Cover} to name a 
few on planar graphs 
\cite{AlberBFKN02,DemaineFHT05talg,FominT06},
 and more generally, on $H$-minor-free graphs 
\cite{DemaineFHT05,DemaineH08,DemaineH07-II}. These algorithms are obtained by showing a combinatorial relation between the parameter and the structure of the input graph and proofs require 
strong graph theoretic arguments. This graph-theoretic and combinatorial component in the design of subexponential time parameterized algorithms makes it of an independent interest.  

Demaine et al.~\cite{DemaineFHT05} abstracted out the ``common theme'' among the parameterized subexponential  time 
algorithms on sparse graphs and created the meta-algorithmic theory of Bidimensionality. The bidimensionality 
theory unifies and improves almost all known previous subexponential algorithms on spare graphs. The theory 
is based on algorithmic and combinatorial extensions to various  parts of Graph Minors Theory of Robertson and 
Seymour~\cite{RobertsonST94} 
and provides a simple criteria for checking whether a parameterized problem is solvable in subexponential time 
on sparse graphs. The theory applies to graph problems that are {\em bidimensional} in the sense that the 
value of the solution for the problem in question on $k\times k$ grid or ``grid like graph'' is  at least $\Omega(k^2)$ 
and the value of solution decreases while contracting or sometime deleting the edges. Problems that are bidimensional include {\sc $k$-Feedback Vertex Set}, {\sc $k$-Edge Dominating Set}, {\sc $k$-Leaf Spanning Tree}, {\sc $k$-Path}, $k$-{\sc $r$-Dominating Set}, {\sc $k$-Vertex Cover} and many others. In most cases 
we obtain subexponential time algorithms for a problem using  bidimensionality theory in following steps.  
Given an instance $(G,k)$ to a bidimensional problem $\Pi$, in polynomial time we either decide that it is an 
yes instance to $\Pi$ or the treewidth of $G$ is $\cO(\sqrt k)$. In the second case, using known constant factor 
approximation algorithm for the treewidth, we find a tree decomposition of width $\cO(\sqrt k)$ for $G$ and then 
solve the problem by doing dynamic programming over the obtained tree decomposition. This approach combined 
with Catalan structure based dynamic programming over graphs of bounded treewidth has led to 
$2^{\cO(\sqrt{k})} n^{\cO(1)}$ time algorithm for {\sc $k$-Feedback Vertex Set}, {\sc $k$-Edge Dominating Set}, {\sc $k$-Leaf Spanning Tree}, {\sc $k$-Path}, $k$-{\sc $r$-Dominating Set}, {\sc $k$-Vertex Cover} and many others 
on planar graphs~\cite{DemaineFHT05talg,DemaineFHT05,DornPBF09} and in some cases like {\sc $k$-Dominating Set} and {\sc $k$-Path} on $H$-minor free graphs~\cite{DemaineFHT05,DornFT08}. We refer to surveys by Demaine and Hajiaghayi~\cite{DemaineH08} and Dorn et al.~\cite{DornFT08-II} for further details 
on bidimensionality and subexponential parameterized algorithms.

While bidimensionality theory is a powerful algorithmic framework  on undirected graphs, it remains unclear how to apply it to problems on directed graphs (or digraphs).  The main reason is that Graph Minor Theory for digraphs is still in a nascent stage and there are no suitable obstruction theorems so far.  For an example, even the first step of the framework does not work easily on digraphs,  as there is no unique notion of 
directed $k\times k$ grid. Given a $k\times k$ undirected grid we can make $2^{\cO(k^2)}$ distinct 
directed grids by choosing orientations for the edges. Hence, unless we can 
guarantee a lower bound of $\Omega(k^2)$ on the size of solution of a problem for {\em any}  
directed  $k\times k$ grid, the bidimensionality theory does not look applicable for problems on digraphs.  Even the analogue  of treewidth for digraphs is not unique and several alternative definitions have been proposed. 
Only recently the first non-trivial subexponential parameterized algorithms on digraphs was 
obtained.  Alon et al.~\cite{AlonLS09} introduced the method of chromatic coding, a variant of color coding~\cite{AlonYZ95}, 
and combined it with divide and conquer to obtain $2^{\cO(\sqrt k \log k)}n^{\cO(1)}$ for {\sc $k$-Feedback Arc Set} 
in tournaments.

\medskip\noindent{\textbf{Our contribution.}} In this paper we make the first step beyond bidimensionality by obtaining subexponential time algorithms for problems on sparse digraphs. We 
develop two different methods to achieve subexponential time parameterized algorithms for digraph problems when the input graph can be embedded on some surface or the underlying undirected graph 
excludes some fixed graph $H$ as a minor. 

 
\noindent{\em \bf Quasi-bidimensionality.} 
Our first technique can be thought of as ``bidimensionality in disguise''. We observe that 
given a digraph $D$, whose underlying undirected graph $UG(D)$ 
excludes some fixed graph $H$ as a minor, if we can 
remove $o(k^2)$ vertices from the given digraph to obtain a digraph 
whose underlying undirected graph has a constant treewidth, then the treewidth of $UG(D)$ is 
$o(k)$.   So given an instance $(D,k)$ to a problem $\Pi$, in polynomial time we either decide that it is an 
yes instance to $\Pi$ or the treewidth of $UG(D)$ is $o(k)$. 
In the second case, as in the framework based on bidimensionality, we solve the problem  
by doing dynamic programming over the tree decomposition of $UG(D)$. The 
dynamic programming part of the framework is problem-specific and runs in time  $2^{o(k)}  + n^{\cO(1)}$. 
We exemplify this technique on a well studied problem of {\sc $k$-Leaf Out-Branching}. 

We say that a subdigraph $T$ on vertex set $V(T)$ of a digraph $D$ on vertex set $V(D)$ 
is an {\em out-tree} if $T$ is an
oriented tree with only one vertex $r$ of in-degree zero (called
the {\em root}). The vertices of $T$ of out-degree zero are called {\em
leaves} and every other vertex is called an {\em internal vertex}. If $T$ is a spanning out-tree, that is, $V(T)=V(D)$, then
$T$ is called an {\em out-branching} of $D$. Now we are in position to define the problem formally. 
\begin{quote}
{\sc $k$-Leaf Out-Branching (\LOB):} Given a digraph $D$ with the vertex set $V(D)$ and the arc set $A(D)$ 
and a positive integer $k$, check whether there exists an out-branching with at least $k$ leaves.   
\end{quote}

The study of {\sc $k$-Leaf Out-Branching} has been at forefront of research in parameterized algorithms in the 
last few years. Alon et al.  
\cite{AlonFGKS_sidma09}  showed that the problem is
 fixed parameter tractable by giving  an algorithm that 
 decides in time $\cO(f(k)n)$ whether a strongly connected digraph has an out-branching
with at least $k$ leaves. Bonsma and Dorn \cite{BonsmaD09} extended this result to all 
digraphs, and improved  the running time of the algorithm. Recently, Kneis et al.~\cite{KLR08} provided a parameterized algorithm solving
the problem in time $4^kn^{\cO(1)}$. This result was further improved to $3.72^{k}n^{\cO(1)}$ 
by Daligaut et al.~\cite{abs-0810-4946}. Fernau et al.~\cite{FernauFLRSV09} showed that for the rooted version of the problem,  
where apart from the input instance we are also given a root $r$ and one asks for a $k$-leaf out-branching rooted at 
$r$, admits a $\cO(k^3)$ kernel. Furthermore they also show that \LOB{} does not admit polynomial kernel 
unless polynomial hierarchy collapses to third level. Finally, Daligault and Thomass{\'e}~\cite{abs-0904-2658} obtained a $\cO(k^2)$ kernel for the rooted version of the \LOB{} problem and gave a constant factor 
approximation algorithm for \LOB{}.

Using our new  technique in combination with kernelization result of~\cite{FernauFLRSV09}, we get 
an algorithm for \LOB{} that runs in time  $2^{\cO(\sqrt k \log k)} n + n^{\cO(1)}$ for digraphs 
whose underlying undirected graph is  $H$-minor-free. For planar digraphs our algorithm runs in 
$2^{\cO(\sqrt k )} n + n^{\cO(1)}$ time. 

\smallskip 

\noindent{\em \bf Kernelization and Divide \& Conquer.}
Our second technique is a combination of divide and conquer, kernelization and dynamic 
programming over graphs  of bounded treewidth. 
Here, using a combination of kernelization and a Baker style layering technique for obtaining polynomial time approximation schemes~\cite{Baker94}, we reduce the instance of a given problem to 
$2^{o(k)}n^{\cO(1)}$ many new instances of the same problem. These new instances have the following properties: (a) the treewidth of the underlying undirected graph of these instances is bounded by $o(k)$; and (b) the original input is an yes instance if and only if at least one of the newly generated instance is. We exhibit this technique on the 
{\sc $k$-Internal Out-Branching} problem, a parameterized version of a generalization of {\sc Directed Hamiltonian Path}. 

\begin{quote}
{\sc $k$-Internal Out-Branching (\IOB{}):} Given a digraph $D$ with the vertex set $V(D)$ and the arc set $A(D)$ and a positive 
integer $k$, check whether there exists an out-branching with at least $k$ internal vertices.   
\end{quote}
Prieto and Sloper~\cite{PrietoS05} studied the {\em undirected} version of this problem and gave an algorithm with running  time $2^{4k\log{k}}n^{\cO(1)}$ and obtained a kernel of size $\cO(k^2)$. Recently, Fomin et al.~\cite{FominGST09} 
obtained a vertex kernel of size $3k$ and gave an algorithm for the undirected version of \IOB{} running in time 
$8^kn^{\cO(1)}$. Gutin et al. \cite{abs-0801-1979} obtained an algorithm of running time 
$2^{\cO(k\log k)}n^{\cO(1)}$ for \IOB{} and gave a kernel of size of $\cO(k^2)$ using the well known method of crown-decomposition. Cohen et al.~\cite{CohenFGKSY09} improved the algorithm for \IOB{} and gave an algorithm with 
running time $49.4^k n^{\cO(1)}$. Here, we  obtain a subexponential time algorithm for \IOB{} 
with running time $2^{\cO(\sqrt k \log k)}  + n^{\cO(1)}$ on directed planar graphs and digraphs whose underlying undirected graphs are apex minor-free. 

Finally, we also observe that for any $\ve>0$, there is an algorithm finding in time   $\cO((1+\ve)^k n^{f(\ve)})$  a directed path of length at least $k$ (the  {\sc $k$-Directed Path} problem) in a digraph which underlying undirected graph excludes a fixed apex graph as a minor. The existence of subexponential parameterized algorithm for this problem remains open.

\section{Preliminaries}
\label{prelim} 
Let $D$ be a digraph. By $V(D)$ and $A(D)$ we
represent the vertex set and arc set of $D$, respectively. 
Given a subset $V'\subseteq V(D)$ of a digraph $D$, let $D[V']$ denote the
digraph induced by $V'$. The {\em underlying graph} $UG(D)$ of $D$
is obtained from $D$ by omitting all orientations of arcs and by
deleting one edge from each resulting pair of parallel edges. 
A vertex $u$ of $D$ is an {\em in-neighbor} ({\em
out-neighbor}) of a vertex $v$ if $uv\in A(D)$ ($vu\in A(D)$,
respectively). The {\em in-degree} $d^-(v)$ ({\em out-degree}
$d^+(v)$) of a vertex $v$ is the number of its in-neighbors
(out-neighbors). We say that a subdigraph $T$ of a digraph $D$ 
is an {\em out-tree} if $T$ is an oriented tree with only one vertex 
$r$ of in-degree zero (called the {\em root}). The vertices of $T$ of out-degree zero are called 
{\em leaves} and every other vertex is called an {\em internal vertex}. If $T$ is a spanning out-tree, that is, $V(T)=V(D)$, then
$T$ is called an {\em out-branching} of $D$.  An out-branching (respectively. out-tree) rooted at $r$ is called {\em $r$-out-branching} 
(respectively. $r$-out-tree). We define the operation of a {\em contraction of a directed arc} as follows. 
An arc $uv$ is contracted as follows: add a new vertex $u'$, and for each arc $wv$ or $wu$ add
the arc $wu'$ and for an arc $vw$ or $uw$ add the arc $u'w$,
remove all arcs incident to $u$ and $v$ and the vertices $u$ and $v$. We  call a loopless digraph $D$ {\em rooted}, if there exists a pre-specified vertex $r$ of in-degree $0$ as a root $r$ and $d^+(r)\geq 2$.  
The rooted digraph $D$ is called {\em connected} if every vertex in $V(D)$ is reachable from $r$ by a directed path.



Let $G$ be an undirected  graph  with the vertex set $V(G)$ and the edge set $E(G)$. 
For a subset $V'\subseteq V(G)$,
by $G[V']$ we mean the subgraph of $G$ induced by $V'$. By $N(u)$
we denote (open) neighborhood of $u$ that is the set of all vertices
adjacent to $u$ and by $N[u]=N(u)\cup \{u\}$.
Similarly, for a subset $D \subseteq V$, we define $N[D]=\cup_{v\in D} N[v]$. 
The {\em diameter} of
a graph $G$, denoted by $diam(G)$, is defined to be the maximum length of a
shortest path between any pair of vertices of $V(G)$. 

Given an edge $e=uv$ of a graph $G$, the
graph $G/e$ is obtained by contracting the edge $uv$; that is, we get
$G/e$ by identifying the vertices $u$ and $v$ and removing all the loops
and duplicate edges. A {\em minor} of a graph $G$ is a graph $H$ that can
be obtained from a subgraph of $G$ by contracting edges.  A graph
class $\mathcal C$ is {\em  minor closed} if any minor of any graph
in $\mathcal C$ is also an element of $\mathcal C$. A minor closed
graph class $\mathcal C$ is $H${\em -minor-free}  or simply
$H${\em -free} if
$H \notin \mathcal C$. A graph $H$ is called an apex graph if the removal 
of one vertex makes it a planar graph.

A {\em tree decomposition} of a (undirected) graph $G$ is a pair
$(X,T)$ where $T$ is a tree whose vertices we will call {\em
nodes} and $X=(\{X_{i} \mid i\in V(T)\})$ is a collection of
subsets of $V(G)$ such that $(a)$ $\bigcup_{i \in V(T)} X_{i} = V(G)$, 
$(b)$  for each edge $vw \in E(G)$, there is an $i\in V(T)$
such that $v,w\in X_{i}$, and $(c)$ 
for each $v\in V(G)$ the set of nodes $\{ i \mid v \in X_{i}
\}$ forms a subtree of $T$. 
The {\em width} of a tree decomposition $(\{ X_{i} \mid i \in V(T)\}, T)$ equals $\max_{i \in V(T)} \{|X_{i}| - 1\}$. The {\em
treewidth} of a graph $G$ is the minimum width over all tree
decompositions of $G$. We use notation $\tw(G)$ to denote the
treewidth of a graph $G$.


A parameterized problem is said to admit a {\it polynomial kernel} if there is a polynomial time algorithm (where the degree of the polynomial is independent of $k$), called a {\em kernelization} algorithm, that reduces the input instance down to an instance with size bounded by a polynomial $p(k)$ in $k$, while preserving the answer. This reduced instance is called a {\em $p(k)$ kernel} for the problem. 
See~\cite{Niedermeierbook06} for an introduction to kernelization.

\section{Method I --  Quasi Bidimensionality}
In this section we present our first approach. In general, a subexponential time algorithm using bidimensionality 
is obtained by showing that the solution for a problem in question is at least $\Omega (k^2)$ on $k\times k$ 
(contraction) grid minor. 
Using this we reduce the problem to a question on graph with treewidth  $o(k)$. 
We start with a lemma which 
enables  us to use the framework of bidimensionality for digraph problems, though not as directly as for 
undirected graph problems. 

\begin{lemma}
\label{lem:qb}
Let $D$ be a digraph such that $UG(D)$ excludes a fixed graph $H$ as a minor. For any constant $c\geq 1$, if there exists a subset $S\subseteq V(D)$ with $|S|=s$ such that 
$\tw(UG(D[V(D)\setminus S]))\leq c$,
  then 
$\tw(UG(D)) = \cO(\sqrt{s})$. 
\end{lemma}
\begin{proof}
By \cite{DemaineH08}, for any $H$-minor-free graph $G$ with treewidth more than $r$, there is a constant $\delta > 1$ only dependent on $H$ such that $G$ has a 
$\frac{r}{\delta} \times \frac{r}{\delta}$ grid minor. Suppose  $\tw(UG(D)) > \delta (c+1) \sqrt{s}$
 then $UG(D)$ contains a $(c+1) \sqrt{s}\times (c+1) \sqrt{s}$ grid as a minor. Notice that this grid minor can not be destroyed by any vertex set $S$ of size at most $s$. 
That is, if we delete any vertex set $S$ with $|S|=s$ from this grid, it will still contain a 
$(c+1)\times (c+1)$ subgrid. Thus, $UG(D[V(D)\setminus S])$ contains a 
$(c+1)\times (c+1)$  grid minor and hence by~\cite[Exercise 11.6]{FlumGrohebook} we have that 
$\tw(UG(D[V(D)\setminus S]))>c$. This shows that we need to delete more than 
$s$ vertices from $UG(D)$ to obtain
a graph with treewidth at most $c$, a contradiction.
\end{proof}

Using Lemma~\ref{lem:qb}, we show that {\sc $k$-Leaf-Out-Branching}  problem has 
a subexponential time algorithm on digraphs $D$ such that $UG(D)$ exclude a fixed graph $H$ as a minor.  
For our purpose a rooted version of \LOB{} will also be useful which we define now. In  the 
{\sc Rooted $k$-Leaf-Out-Branching (\rLOB)} problem apart from $D$ and $k$ 
the root $r$ of the tree searched for is also a part of the input and the objective is to check whether there exists an  
$r$-out-branching with at least $k$ leaves. We now state our main combinatorial lemma and postpone its proof for a while. 

\begin{lemma}
\label{lem:struct}
Let $D$ be a digraph such that $UG(D)$ excludes a fixed graph $H$ as a minor, $k$ be a positive integer and 
$r\in V(D)$ be the root. Then in polynomial time either we can construct an $r$-out-branching with at least $k$ leaves in $D$ 
or find a digraph $D'$ such that following holds. 
\begin{itemize}
\setlength{\itemsep}{-3pt}
\item 
$UG(D')$ excludes the fixed graph $H$ as a minor; 
\item 
$D$ has an $r$-out-branching with at least $k$ leaves if and only if $D'$ has an $r$-out-branching with at least $k$ leaves; 
\item 
there exists a 
subset $S\subseteq V(D')$ such that $|S| =\cO(k)$ and $\tw(U(D'[V(D')\setminus S])\leq c$, 
$c$ a constant. 
 \end{itemize}
 \end{lemma}
 
Combining Lemmata~\ref{lem:qb} and~\ref{lem:struct} we obtain the following result. 
\begin{lemma}
\label{lemma:twbound}
Let $D$ be a digraph such that $UG(D)$ excludes a fixed graph $H$ as a minor, $k$ be a positive integer and 
$r\in V(D)$ be a root. Then in polynomial time either we can construct an $r$-out-branching with at least $k$ leaves in $D$  or find a digraph $D'$ such that $D$ has an 
$r$-out-branching with at least $k$ leaves if and only if $D'$ has an $r$-out-branching with at least $k$ leaves. Furthermore $\tw(UG(D'))= \cO(\sqrt k)$. 
 \end{lemma}
 

When a tree decomposition of $UG(D)$ is given, dynamic programming methods can be used to decide whether $D$ has an out-branching with at least $k$ leaves, 
 see ~\cite{abs-0801-1979}.
The time complexity of such a procedure is $2^{\cO(w\log w)} n$, where $n=|V(D)|$ and $w$ is the width of the tree decomposition. Now we are ready to prove the main theorem of this section 
assuming the combinatorial Lemma \ref{lem:struct}. 

\begin{theorem}
\label{thm:LOB}
The \LOB{} problem can be solved in time $2^{\cO(\sqrt k \log k)} n + n^{\cO(1)}$ on digraphs with $n$ vertices such that the underlying undirected graph excludes a fixed graph $H$ as a minor. 
\end{theorem}
\begin{proof}
Let $D$ be a digraph where $UG(D)$ excludes a fixed graph $H$ as a minor. We guess a vertex $r\in V(D)$ as a root. This only adds a factor of $n$ to our algorithm. By Lemma~\ref{lemma:twbound}, we can either compute, in polynomial time, an $r$-out-branching with at least $k$ leaves in $D$ or find a digraph $D'$ with $UG(D')$ excluding a fixed graph $H$ 
as a minor and $\tw(UG(D'))=\cO(\sqrt k)$. In the later case, using the constant factor approximation algorithm of Demaine et al.~\cite{DemaineHK05} for computing the treewidth of a $H$-minor free graph,  
we find a tree decomposition  of width $\cO(\sqrt k)$ for $UG(D')$ in time $n^{\cO(1)}$. With the previous observation that we can find an $r$-out-branching with at least $k$ leaves, if exists one, 
in time $2^{\cO(\sqrt k \log k)} n$ using dynamic programming over graphs of bounded treewidth, we have that we can solve \rLOB{} in time $2^{\cO(\sqrt k \log k)}   n^{\cO(1)}$. Hence, we need 
$2^{\cO(\sqrt k \log k)}n^{\cO(1)}$ to solve the \LOB{} problem. 

To obtain the claimed running time bound we use  the known kernelization algorithm after we have guessed the root $r$. 
Fernau et al.~\cite{FernauFLRSV09}  
gave an $\cO(k^3)$ kernel for \rLOB{} which preserves the graph class. That is, given an instance $(D,k)$ of \rLOB{}, in polynomial time they output an equivalent instance $(D'',k)$ of \rLOB{} 
such that (a) if $UG(D)$ is $H$-minor free then so is $UG(D'')$; and (b) $|V(D'')| =\cO(k^3)$. 
We will use this kernel for our algorithm rather than the $\cO(k^2)$ kernel for \rLOB{} obtained by Daligault and Thomass{\'e}~\cite{abs-0904-2658}, as they do not preserve the graph class. 
So after we have guessed the root $r$, we obtain an equivalent instance $(D'',k)$ for \rLOB{} using the kernelization procedure described in~\cite{FernauFLRSV09}.  Then using the algorithm described in the 
previous paragraph we can solve  \rLOB{} in time $2^{\cO(\sqrt k \log k)}  + n^{\cO(1)}$. Hence, we need 
$2^{\cO(\sqrt k \log k)}n  + n^{\cO(1)}$ to solve \LOB{}. 
\end{proof}

Given a tree decomposition of width $w$ of $UG(D)$ for a planar digraph $D$, we can solve  \LOB{} 
using dynamic programming methods  
in time $2^{\cO(w)} n$. This brings us to the following theorem. 
 
\begin{theorem}$[\star]$\footnote{The proofs marked with [$\star$] have been moved to the appendix due to space restrictions.}
\label{lemma:planarLOB}
 The \LOB{} problem can be solved in time $2^{\cO(\sqrt k)}n + n^{\cO(1)}$ on digraphs with $n$ vertices when the underlying undirected graph is planar. 
 \end{theorem}
 
 
 
\subsection{Proof of Lemma~\ref{lem:struct}}
To prove the combinatorial lemma we need a few recent results 
from the literature on out-branching problems. We start with some definitions given in~\cite{abs-0904-2658}.  A {\em cut} of 
$D$ is a subset $S$ such that there exists a vertex $z\in V(D)\setminus S$ such that $z$ is not reachable from $r$ in $D[V(D)\setminus S]$. 
We say that $D$ is \emph{$2$-connected} if there exists no cut of size one in $D$ or equivalently there are at least two vertex disjoint paths from $r$ to 
every vertex in $D$. 

\begin{lemma}
[\cite{abs-0904-2658}]
\label{stephanone}
Let $D$ be a rooted $2$-connected digraph with $r$ being its root. Let $\alpha$ be the number of vertices in $D$ with in-degree at 
least $3$. Then $D$ has an out-branching rooted at $r$ with at least $\alpha/6$ 
leaves and such an out-branching can be found in polynomial time.
\end{lemma}

A vertex $v \in V(D)$ is called a {\em nice vertex} if $v$ has an in-neighbor which is not 
its out-neighbor. The following lemma is proved in~\cite{abs-0904-2658}. 
\begin{lemma}
[\cite{abs-0904-2658}]
\label{stephantwo}
Let $D$ be a rooted $2$-connected digraph rooted at a vertex $r$. 
Let $\beta$ be the number of nice vertices in $D$. Then $D$ has an out-branching rooted at $r$ with at least $\beta/24$ leaves 
and such an out-branching can be found in polynomial time.  
\end{lemma}

\begin{proof}[{\bf Proof of Lemma~\ref{lem:struct}}]
To prove the combinatorial lemma, we consider two cases based on whether or not $D$ is $2$-connected.

\noindent 
{\bf Case 1)} {\sl $D$ is a rooted $2$-connected digraph.} 

We prove this case in the following claim.
\begin{claim}
\label{claim:2connected}
Let $D$ be a rooted $2$-connected digraph with root $r$ and a positive integer $k$. Then in polynomial time, we 
can find an out-branching rooted at $r$ with at least $k$ leaves or find a set $S$ of at most $30k$ vertices
whose removal results in a digraph whose underlying undirected graph has treewidth one.
\end{claim}
\begin{proof}
If $\alpha \geq  6k$, then we are done by Lemma \ref{stephanone}. Similarly if 
$\beta \geq 24k$, then we are done by Lemma \ref{stephantwo}. Hence we assume
that $\alpha < 6k$ and $\beta < 24k$. Let $S$ be the set of nice vertices 
and vertices of in-degree at least $3$ in $G$. Then $|S| < \alpha+\beta \leq 30k$. Observe that 
$D[V(D)\setminus S]$ is simply a collection of directed paths where every edge of the path is a directed $2$-cycle. 
This is because $D[V(D)\setminus S]$ contains only those vertices which are not nice (that is, those vertices whose 
in-neighbors are also out-neighbors) and are of in-degree at most two. 
Hence, if there is an arc $xy$ in $D[V(D)\setminus S]$, then the arc $yx$ also exists in $D[V(D)\setminus S]$. 
Next we note that $D[V(D)\setminus S]$ does not contain a directed cycle of length more than 
two. We prove the last assertion as follows. Let $\cal C$ be a directed cycle in $D[V(D)\setminus S]$ of length at least 
$3$. Since $D$ is a rooted $2$-connected digraph, we have a vertex $v$ on the cycle $\cal C$ such that there is  
a path from $r$ to $v$ without using any other vertex from the cycle $\cal C$. This implies that the in-degree of $v$ is 
at least $3$ in $D$ and hence $v\in S$, contrary to our assumption that $v\notin S$. This proves that 
$D[V(D)\setminus S]$ does not contain a directed cycle of length more than two.
Hence the underlying undirected graph $UG(D[V(D)\setminus S])$ is just a collection of paths 
and hence $\tw(UG(D[V(D)\setminus S]))$ is one.
\end{proof}
\noindent 
{\bf Case 2)} {\sl $D$ is not $2$-connected.}

Since $D$ is not $2$-connected, it has cut vertices, 
those vertices that separate $r$ from some other vertices. 
We deal with the cut vertices in three cases. 
Let $x$ be a cut vertex of $D$.  The three cases we consider are following. 

\noindent 
{\bf Case 2a)} {\sl There exists an arc $xy$ that disconnects at least two vertices from $r$.} 

In this case, we {\em contract the arc $xy$}.  After repeatedly applying Case $2a)$, we obtain a digraph $D'$ such that any arc out of a cut vertex $x$ of $D'$ disconnects at most $1$ vertex. The resulting 
digraph $D'$ is the one mentioned in the Lemma. Since
we have only contracted some arcs iteratively to obtain $D'$, it is clear
that $UG(D')$ also excludes $H$ as a minor. 
The proof that such contraction does not decrease the number of leaves follows from a reduction rule given in~\cite{FernauFLRSV09}. We provide a proof for completion.

\begin{claim}$[\star]$
\label{claim:contract}
Let $D$ be a rooted connected digraph with root $r$, let $xy$ be an arc that disconnects at least two vertices from $r$ and $D'$ be the digraph obtained after contracting the arc $xy$.   
 Then $D$ has an $r$-out-branching with at least $k$ leaves if and only if $D'$ has an $r$-out-branching with at least $k$ leaves. 
\end{claim}


Now we handle the remaining cut-vertices of $D'$ as follows.  Let $\cal S$ be the set 
of cut vertices in $D'$. For every vertex $x\in {\cal S}$, we associate {\em a cut-neighborhood} $C(x)$, which 
is the set of out-neighbors of $x$ such that there is no path from $r$ to any vertex in $C(x)$ in 
$D'[V(D')\setminus \{x\}]$. By $C[x]$ we denote $C(x)\cup \{x\}$. 
The following observation is used to handle other cases. 
\begin{claim}
\label{claim:intersect}
 Let $\cal S$ be the set of cut vertices in $D'$. Then for every pair of vertices $x,y\in {\cal S}$ and $x\neq y$,  
 we have that  $C[x]\cap C[y]=\emptyset$. 
\end{claim}
\begin{proof}
To the contrary let us assume that $C[x] \cap C[y] \neq \emptyset $. We note that $C[x]\cap C[y]$ can only have a vertex $v \in \{x,y\}$. To prove this, assume to the contrary 
that we have a vertex $v\in C[x]\cap C[y]$ and $v\notin \{x,y\}$. But then it contradicts 
the fact that $v \in C[x]$, as $x$ doesn't separate $v$ from $r$ due
to the path between $r$ and $v$ through $y$. Thus, either $x\in C(y)$ or 
$y\in C(x)$.  Without loss of generality let $y\in C(x)$. This implies that we have an arc $xy$ and 
there exists a vertex $z\in C(y)$ such that $z\notin C(x)$. But then the arc $xy$ 
disconnects at least two vertices $y$ and $z$ from $r$ and hence Case $2a$ would have applied. This proves the claim. 
\end{proof}

Now we distinguish cases based on cut vertices having cut-neighborhood of size at least $2$ or $1$. Let 
${\cal S}_{\geq 2}$ and ${\cal S}_{=1}$ be the subset of cut-vertices of $D'$ having at least two cut-neighbors and exactly 
one neighbor respectively. 

\noindent
{\bf Case 2b)} {\sl ${\cal S}_{\geq 2} \neq \emptyset$.}

We first bound $|{\cal S}_{\geq 2}|$. Let $A_c=\{xy~|~x\in {\cal S}_{\geq 2}, y\in C(x) \}$ be the set of out-arcs emanating from 
the cut vertices in ${\cal S}_{\geq 2}$ to its cut neighbors. We now prove the following structural claim which is useful for bounding the size of ${\cal S}_{\geq 2}$. 
\begin{claim}
\label{claim:allarc}$[\star]$ If $D'$ has an $r$-out-branching  $T'$ with at least $k$ leaves then $D'$ has an $r$-out-branching  $T$ with at least 
$k$ leaves and containing all the arcs of $A_c$, that is, $A_c\subseteq A(T)$. Furthermore such an out-branching can be found in polynomial time. 
\end{claim}


We know that in any out-tree, the number of internal vertices of out-degree at 
least $2$ is bounded by the number of leaves. Hence if 
 $|{\cal S}_{\geq 2}|\geq k$ then we obtain an $r$-out-branching  $T$ of $D'$  with at least $k$ leaves using 
 Claim~\ref{claim:allarc} and we are done. So from now onwards we assume that  $|{\cal S}_{\geq 2}|=\ell \leq k-1$.  

We now do a transformation to the given digraph $D'$. For every vertex $x\in {\cal S}_{\geq 2}$, we introduce an {\em imaginary} vertex $x^i$ and add an arc $ux^i$ if there is an arc $ux\in A(D')$ and add an arc $x^iv$ if there is an arc 
 $xv\in A(D')$. Basically we duplicate the vertices in  ${\cal S}_{\geq 2}$. Let the transformed graph be called 
 $D^{dup}$. We have the following two properties about $D^{dup}$. First, no vertex in ${\cal S}_{\geq 2}\cup \{x^{i}|x\in{\cal S}_{\geq 2}\} $ is a cut vertex in $D^{dup}$. We sum up the second property in the following claim. 
 \begin{claim}
 \label{claim:dup}
The digraph $D'$ has an $r$-out-branching  $T$ with at least $k$ leaves if and only if $D^{dup}$ has an $r$-out-branching  $T'$ 
with at least $k+\ell$ leaves. 
 \end{claim}
 \begin{proof}
 Given an $r$-out-branching $T$ of $D'$ with at least $k$ leaves, we obtain an out-branching $T'$ of $D^{dup}$ 
 with at least $k+\ell$ leaves by adding an arc $xx^i$ to $T$ for every $x\in {\cal S}_{\geq 2}$. Since every vertex of 
 ${\cal S}_{\geq 2}$ is an internal vertex in $T$, this process only adds $\{x^i~|~x \in {\cal S}_{\geq 2}\}$ 
 as leaves and hence we have at least $k+\ell$ leaves in $T'$. 
  
For the converse, assume that $D^{dup}$ has an $r$-out-branching  $T'$  with at least $k+\ell$ leaves.  
First, we modify the out-branching so that not both of $x$ and $x^i$ 
are internal vertices and we do not lose any leaf. This can be done easily by making
all out arcs in the out-branching from $x$ and making $x^i$ a leaf. That is, if $N_{T'}^{+}(x^i)$ is the set of out-neighbors of 
$x^i$ in $T'$ then we delete the arcs $x^iz$, $z\in N_{T'}^{+}(x^i)$ and add $xz$ for all $z\in N_{T'}^{+}(x^i)$. 
This process can not decrease the number of leaves. Furthermore we can always assume 
that if exactly one of $x$ and $x^i$ is an internal vertex, then $x$ is the internal vertex in $T'$. Now delete 
all the vertices of $\{x^i~|~x \in {\cal S}_{\geq 2}\}$ from $T'$ and obtain $T$.  Since the vertices in the set 
$\{x^i~|~x \in {\cal S}_{\geq 2}\}$ are leaves of $T'$, we have that $T$ is an $r$-out-branching in $D'$. Since in the whole 
process we have only deleted $\ell$ vertices we have that $T$ has at least $k$ leaves. 
 \end{proof}
Now we move on to the last case. 

\noindent
{\bf Case 2c)} ${\cal S}_{=1} \neq \emptyset$. 

Consider the arc set 
 $A_{p}=\{xy~|~x\in {\cal S}_{=1}, y\in C(x)\}$. Observe that $A_p\subseteq A(D')\subseteq A(D^{dup})$ and 
 $A_p$ forms a {\em matching} in $D^{dup}$ because of Claim~\ref{claim:intersect}. Let $D^{dup}_{c}$ be the 
 digraph obtained from $D^{dup}$ by contracting the arcs of $A_p$.  That is, for every arc $uv\in A_p$, 
 the contracted graph is obtained by identifying the vertices $u$ and $v$ as $uv$ and removing all the loops and duplicate arcs.
 \begin{claim}
 \label{claim:dupc}
 Let  $D^{dup}_{c}$ be the digraph obtained by contracting the arcs of $A_p$ in  $D^{dup}$. Then the following holds. 
 \begin{enumerate}
 \setlength{\itemsep}{-2pt}
 \item The digraph $D^{dup}_{c}$ is $2$-connected;
 \item If $D^{dup}_{c}$ has an $r$-out-branching $T$ with at least $k+\ell$ leaves then $D^{dup}$ has an $r$-out-branching 
 with at least $k+\ell$ leaves.  
 \end{enumerate}
 \end{claim}
 \begin{proof}
 The digraph $D^{dup}_{c}$ is $2$-connected by the construction as we have iteratively removed all cut-vertices. If $D^{dup}_{c}$ has an $r$-out-branching $T$ with at least $k+\ell$ leaves then we can obtain a $r$-out-branching 
 with at least $k+\ell$ leaves for $D^{dup}$ by expanding  each of the contracted vertices to arcs in $A_p$. 
  \end{proof}
We are now ready to combine the above claims to complete the proof of the lemma. We first apply 
Claim~\ref{claim:2connected} on $D^{dup}_{c}$ with $k+\ell$. Either we get an $r$-out-branching $T'$ with at least 
$k+\ell$ leaves or a set $S'$ of size at most $30(k+\ell)$  such that $\tw(UG(D^{dup}_c[V(D^{dup}_c)\setminus S]))$ is 
one. In the first case, by Claims~\ref{claim:dup} and \ref{claim:dupc} we get an $r$-out-branching $T$ with at least $k$ 
leaves in $D'$. In the second case we know that there is a vertex set $S'$ of size at most $30(k+\ell)$  such that
 $\tw(UG(D^{dup}_c[V(D^{dup}_c)\setminus S']))$ is one. Let $S^*=\{u~|~uv\in S', vu \in S', u\in S'\}$ be the set of vertices 
 obtained from $S'$ by expanding the contracted vertices in $S'$. Clearly the size of 
 $|S^*|\leq  2|S'|\leq 60 (k+\ell)\leq 120 k =\cO(k). $ 
 We now show that the treewidth of the underlying undirected graph of $D^{dup}[V(D^{dup})\setminus S^*]$ is at most
  $3$. This follows from the observation that  $\tw(UG(D^{dup}_c[V(D^{dup}_c)\setminus S']))$ is one. Hence given a
  tree-decomposition of width one for $UG(D^{dup}_c[V(D^{dup}_c)\setminus S'])$ we can obtain a tree-decomposition 
  for $UG(D^{dup}[V(D^{dup})\setminus S^*])$ by expanding the contracted vertices. This can only double the bag 
  size and hence the treewidth of $UG(D^{dup}[V(D^{dup})\setminus S^*])$ is at most $3$, as the bag size can at most be 
  $4$. Now we take $S=S^*\cap V(D')$ and since $V(D')\subseteq V(D^{dup})$, we have that 
  $tw(UG(D[V(D)\setminus S]))\leq 3$. 
This concludes the proof of the lemma. 
\end{proof}

%
\section{Method II - Kernelization and Divide \& Conquer}
In this section we exhibit our second method of designing subexponential time algorithms for 
digraph problems through the {\sc $k$-Internal Out-Branching} problem. In this method 
we utilize the known polynomial kernel for the problem and obtain a collection of 
$2^{o(k)}$ instances such that the input instance is an ``yes'' instance if and only if one of the instances in our collection is. 
The property of the instances in the collection which we make use of is that the treewidth of the underlying 
undirected graph of these instances is $o(k)$. The last property brings dynamic programming on graphs of bounded 
treewidth into picture as the final step of the algorithm. 

Here, we will solve a rooted version of the \IOB{} problem, called  
{\sc Rooted $k$-Internal Out-Branching (\rIOB{})}, where apart from $D$ and 
$k$ we are also given a root $r\in V(D)$, and the objective is to find an $r$-out-branching, if exists one, with at least $k$ internal vertices. 
The \IOB{} problem can be reduced to \rIOB{} by guessing the root $r$ at the additional cost of $|V(D)|$ in the running time of 
the \rIOB{} problem. Henceforth, we will only consider \rIOB{}. We call an 
$r$-out-tree $T$ with $k$ internal vertices {\em minimal} if deleting any leaf results in an $r$-out-tree with at most $k-1$ internal vertices. 
A well known result relating minimal $r$-out-tree $T$ with $k$ internal vertices with a solution to \rIOB{} is as follows. 
\begin{lemma}[\cite{CohenFGKSY09}]
\label{lem:iobminimal}
Let $D$ be a rooted connected digraph with root $r$. 
 Then $D$ has an $r$-out-branching $T'$ 
with at least $k$ internal vertices if and only if $D$ has a minimal $r$-out-tree $T$ with $k$ internal vertices 
with $|V(T)|\leq 2k-1$.  Furthermore, given a minimal $r$-out-tree $T$, we can find an $r$-out-branching $T'$ with at least $k$ internal vertices in 
polynomial time. 
\end{lemma}

We also need another known result about kernelization for \IOB{}. 
\begin{lemma}[\cite{abs-0801-1979}]
\label{lem:iobkernel}
{\sc $k$-Internal Out-Branching} admits a polynomial kernel of size $8k^2+6k$. 
\end{lemma}
In fact, the kernelization algorithm presented in~\cite{abs-0801-1979} works 
for all digraphs and has a unique reduction rule which only {\em deletes 
vertices}. This implies that if we start with a graph $G \in \mathscr G$ where  $\mathscr G$ excludes a fixed graph $H$ as a minor, then the graph $G'$ obtained after applying kernelization algorithm still belongs to $\mathscr G$. 

Our algorithm tries to find a minimal $r$-out-tree $T$ with $k$ internal vertices with $|V(T)|\leq 2k-1$ 
recursively. As the first step of the algorithm we obtain a set of 
$2^{o(k)}$ digraphs such that the underlying undirected graphs have
treewidth $\cO(\sqrt k)$, and the original problem is a ``yes'' instance if and only
at least one of the $2^{o(k)}$ instances is a ``yes'' instance.
More formally, we prove the following lemma. 

\begin{lemma}
\label{lem:internaltree1}
Let $H$ be a fixed apex graph and ${\mathscr G}$ be a minor closed graph class excluding $H$ as a minor. Let $(D,k)$ be an instance to 
{\sc $k$-Internal Out-Branching} such that $UG(D)\in {\mathscr G}$. Then 
there exists a collection 
$${\cal C}=\Big\{(D_i,k',r)~|~D_i \mbox{ is a subgraph of }D, k'\leq k, r\in V(D), 1\leq i\leq {8k^2+6k \choose \sqrt k}\Big\}, $$
of instances such that $\tw(UG(D_i))= \cO(\sqrt k)$ for all $i$ and $(D,k)$ has 
an out-branching with at least $k$ internal vertices
if and only if there exists an $i$, $r$ and $k'\leq k$ such that $(D_i,k',r)$ has an $r$-out-branching with at least $k'$ internal vertices. 
\end{lemma}
\begin{proof}
The idea of the proof is to do Baker style layering technique~\cite{Baker94} combined with kernelization. 
In the first step we apply the kernelization algorithm given by 
Lemma~\ref{lem:iobkernel} on $(D,k)$ and obtain an equivalent instance $(D',k')$ where $|D'|\leq 8k^2+6k$ and $k'\leq k$ 
for \IOB{}. From now onwards we will confine ourselves to $(D',k')$. Observe that since the kernelization algorithm only 
deletes vertices to obtain the reduced instance from the input digraph, we have that $UG(D')\in {\mathscr G}$. 

Now we reduce the \IOB{} problem to the \rIOB{} problem by guessing a vertex $r\in V(D')$ as a root. Furthermore we try to 
find a minimal $r$-out-tree $T$ with $k'$ internal vertices with $|V(T)|\leq 2k'-1$.  This suffices for our purpose if we know 
that every vertex in $V(D')$ is reachable from the root $r$, as in this case 
Lemma~\ref{lem:iobminimal} is applicable. 

We start with a BFS starting at the vertex $r$ in $UG(D')$. Let the layers created by
doing BFS on $r$ be $L_0^r,L_1^r,\ldots,L_t^r$. If $t \leq \lceil{\sqrt k }\rceil$, then the collection ${\cal C}_r$ consists of $(D',k',r)$.  
For $t\leq \sqrt{k}$, the fact that $\tw(D')=\cO(\sqrt{k})$ follows from the comments later in the proof. Hence from now onwards we assume that $t>\lceil{\sqrt k }\rceil$. Now we partition the
vertex set into $(\lceil \sqrt k \rceil)+1$ parts where the $q$-th part contains
all vertices which are at a distance of $q+i(\lceil{\sqrt k}\rceil)$ from $r$ for various values of $i$. That is, 
let $V(D')=\cup_{q}P_q$, $q\in \{0,\ldots, \lceil{\sqrt k }\rceil\}$. We define
$ P_q=\bigcup L_{q+i(\lceil{\sqrt k }\rceil+1)}^r,~ i \in \left\{0,\ldots,
\left\lfloor \frac{t-\sqrt k}{\lceil{\sqrt k }\rceil+1}\right\rfloor \right\}. $
It is clear from the definition of $P_q$ that it partitions the vertex set $V(D')$. 
If the input is an ``yes''
 instance then there exists a partition $P_a$ such that it contains at most $\Big\lceil \frac{2k'-1}{\lceil \sqrt k \rceil}\Big\rceil\leq 2 \sqrt k$  
vertices of the minimal $r$-out-tree $T$ we are seeking for. We guess the partition $P_a$ and obtain the collection 
$${\cal C}_r(P_a)=\left\{(D'[V(D')\setminus P_a\cup Z],k',r)~\Big|~Z\subseteq P_a, |Z|\leq 2\sqrt k\right\}.$$
We now claim that for every $Z\subseteq P_a, |Z|\leq 2\sqrt k$, $\tw(UG(D'[V(D')\setminus P_a\cup Z])) = \cO(\sqrt k)$. 
Let $V'=V(D')\setminus P_a$ be the set of vertices after removal of $P_a$ from the vertex set of $D'$. Let the resultant underlying undirected graph be $G'=UG(D'[V'])$ 
with connected components $C_1,\ldots,C_\ell$. We show that each connected component $C_i$ of $G'$ has $\cO(\sqrt k)$ treewidth.  
More precisely, every connected component $C_i$ of $G'$ is a subset of at most $\lceil \sqrt k \rceil + 1$ consecutive layers of the BFS starting at $r$. If we start with $UG(D')$, and delete all BFS layers after these layers and contract all BFS layers before these layers into a single vertex $v$, we obtain a minor $M$ of $UG(D')$. This minor $M$ has diameter at most $\lceil \sqrt k \rceil +2$ and contains $C_i$ as an induced subgraph. Since $UG(D')\in{\mathscr G'}$, we have that $M\in {\mathscr G}$. Furthermore, Demaine and Hajiaghayi~\cite{DemaineH04} have shown that for any fixed apex graph $H$, every $H$-minor-free graph 
of diameter $d$ has treewidth $\cO(d)$. This implies that the $\tw(C_i)\leq \tw(M)\leq \cO(\sqrt k)$. Notice that since every connected component of $G'$ has treewidth $\cO(\sqrt k)$, $G'$ 
itself has $\cO(\sqrt k)$ treewidth. Given a tree-decomposition of width $\cO(\sqrt k)$ for $G'$, we can obtain a tree-decomposition of width $\cO(\sqrt k)$ for $UG(D'[V(D')\setminus P_a\cup Z])$ by adding $Z$ to 
every bag. The collection ${\cal C}_r$  is given by $\cup_{a=0}^{\lceil \sqrt k \rceil} {\cal C}_r(P_a)$. Finally the collection 
${\cal C}=\cup_{r\in V(D')}{\cal C}_r$.

By the pigeon hole principle we know that if $(D',k')$ is an yes instance then there exists a $P_a$ containing at most $2\sqrt k$ vertices of the minimal tree $T$ we are looking for. Since we have run through all $r\in V(D')$ as a potential 
root as well as all subsets of size at most $2\sqrt k$ as the possible intersection of $V(T)$ with $P_a$, we know that
$(D',k')$ has 
an out-branching with at least $k$ internal vertices
if and only if there exists an $i$, $r$ and $k'\leq k$ such that $(D_i,k',r)\in {\cal C}$ has a $r$-out-branching with at least $k'$ internal vertices. This concludes the proof of the lemma.
\end{proof}

Given a tree decomposition of width $w$ for $UG(D)$, one can solve \rIOB{} in time $2^{\cO(w\log w)} n$ using a 
dynamic programming over graphs of bounded treewidth as described 
in~\cite{abs-0801-1979}. 
This brings us to the main theorem of this section. 

\begin{theorem}
The \IOB{} problem can be solved in time $2^{\cO(\sqrt k \log k)} + n^{\cO(1)}$  
 on digraphs with $n$ vertices such that the underlying undirected graph excludes a fixed apex graph $H$ as a minor. 
\end{theorem}
\begin{proof}
As the first step of the algorithm we apply Lemma~\ref{lem:internaltree1} and obtain collection $\cal C$ such that for every $(D,k,r)\in \cal C$, $\tw(UG(D))\in \cO(\sqrt k)$.  
Then using the constant factor approximation algorithm of Demaine et al.~\cite{DemaineHK05} for computing the treewidth of a $H$-minor free graph,  
we find a tree decomposition  of width $\cO(\sqrt k)$ for $UG(D)$ in time $k^{\cO(1)}$. Finally, we apply dynamic programming algorithm 
 running in time $(\sqrt k)^{\cO(\sqrt k)}= 2^{\cO(\sqrt k\log k)}$ on each instance in $\cal C$. If for any of them we get an 
yes answer we return ``yes'', else we return ``no''.  The running time of the algorithm is bounded by 
$$|{\cal C}|\cdot  2^{\cO(\sqrt k\log k)}+n^{\cO(1)}= 2^{\cO(\sqrt k\log k)}\cdot 2^{\cO(\sqrt k\log k)}+n^{\cO(1)} =
2^{\cO(\sqrt k\log k)} +n^{\cO(1)}. $$
We have an additive term of $n^{\cO(1)}$ as we apply the algorithm only on the $\cO(k^2)$ size kernel. This completes the proof. 
\end{proof}

\section{Conclusion and Discussions}
We have given the first subexponential parameterized algorithms on planar digraphs and on the class of digraphs whose underlying undirected graph excludes a fixed graph $H$ 
or an apex graph as a minor. We have outlined two general techniques, and have illustrated them on two well studied problems concerning oriented spanning trees (out branching)--- one that maximizes the number of leaves and the other that maximizes the number of internal vertices. One of our techniques uses the grid theorem on
$H$-minor graphs, albeit in a different way than how it is used on undirected graphs. The other uses Baker type layering technique combined with kernelization and solves the problem on
a subexponential number of problems whose instances have sublinear treewidth. 

We believe that our techniques will be widely applicable and it would be interesting to find other problems where such subexponential algorithms are possible.  
Two famous open problems in this context are whether the {\sc $k$-Directed Path} problem (does a digraph contains a directed path of length at least $k$) and the {\sc $k$-Directed Feedback Vertex Set} problem
(does a digraph can be turned into acyclic digraph by removing at most $k$ vertices)
have subexponential algorithms (at least) on planar digraphs.
However, for the  {\sc $k$-Directed Path} problem, we can reach ``almost" subexponential  running time. More precisely, we have the following theorem. 

\begin{theorem}\label{thm:k-path}$[\star]$
For any $\varepsilon>0$, there is $\delta $ such that the {\sc $k$-Directed Path} problem 
 is solvable in time $\cO((1+\ve)^k \cdot n^\delta)$
on digraphs with $n$ vertices such that the underlying undirected graph excludes a fixed apex graph $H$ as a minor. 
\end{theorem}
Let use remark that similar $\cO((1+\ve)^k n^{f(\ve)})$ results  can also be  obtained  for many other problems including \textsc{Planar Steiner Tree}. 
{\footnotesize{

}}
\newpage

\section{Appendix}
\subsection{Proof of Theorem~\ref{lemma:planarLOB}}
\begin{proof}
We only give an outline of dynamic programming algorithm for planar digraphs that given a tree-decomposition of 
width $w$ decides whether $D$ has an out-branching with at least $k$ leaves in time $2^{\cO(w)}n$. The rest of the 
proof is same as Theorem~\ref{thm:LOB}. 

{\bf Tree collections.} 
Let $G$ be an undirected graph with edge set $E(G)$ and let $E' \subseteq E(G)$. Let $S \subseteq V(G)$ be a vertex set separating $E'$ from $E(G) \setminus E'$, that is, $S$ contains all vertices incident to at least one edge of $E$ and at least one edge of $E(G) \setminus E'$.
 We consider a forest $\mathcal{F}$ with disjoint trees on edges of $E'$ and each intersecting at least one vertex of $S$.
 Let us denote the collection of all such forests $\mathcal{F}$ by $\text{\bf forests}_{E'}(S)$. 
 
  We define an equivalence relation $\sim$ on $\text{\bf forests}_{E'}(S)$ as: for two forests $\mathcal{F}_1, \mathcal{F}_2 \in \mathcal{F}$, $\mathcal{F}_1 \sim \mathcal{F}_2$ if there is a bijection $\mu: \mathcal{F}_1 \rightarrow \mathcal{F}_2$ such that for every tree $T \in {\bf F}$ we have that $ T \cap S = \mu(T) \cap S$. Let $\text{\bf q-forests}(S)$ denote the cardinality of both, the quotient set of  $\text{\bf forests}_{E'}(S)$  plus the quotient set of  $\text{\bf forests}_{E'\setminus E(G)}(S)$ by relation $\sim$.
 In general, $\text{\bf q-forests}(S) \leq |S|!$.
In~\cite{DornPBF09}, the authors show for a planar graph $G$  of treewidth $w$ how to  decompose $G$ by separators of size $\cO(w)$, such that for each such separator $S$,   $\text{\bf q-forests}(S)$ is bounded by $2^{\cO(w)}$. These \emph{branch decompositions} are very closely related to tree decompositions with width parameters bounding each other by constants. Thus, we can simply talk about tree decompositions with some additional structure. 

In this case we use standard 
dynamic programming on such tree decompositions $(X,T)$ (see e.g. \cite{FominT06})
 At every step of dynamic programming for each node of $T$ , we  keep 
track of all the ways the required out-branching can cross the separator $S$ represented by $X$. In other words, we 
count all the ways parts of the out-branching can be routed 
through $E$.
 In the underlying undirected graph,  this is proportional 
to $\text{\bf q-forests}(S)$. 
Since an out-branching is rooted, every subtree is rooted, too. Thus, the only overhead in the directed case compared to the undirected is that we have to guess for each tree $T_F$ in $\mathcal{F}$ if its root is in $S$. In this case, we guess which of the vertices of $T_F \cap S$ is the root. The number of guesses is bounded by $2^{\cO(w)}$ and hence the 
dynamic programming algorithm runs in time $\cO(2^{\cO(w)}n)$. 
\end{proof}

\subsection{Proof of Claim~\ref{claim:contract}}
\begin{proof}  Let the arc $xy$ disconnect at least two vertices $y$ and $w$ from $r$ and let $D'$ be the digraph obtained from $D$ by contracting the arc $xy$. Let $T$ be an $r$-out-branching of $D$ 
with at least $k$ leaves. Since every path from $r$ to $w$ contains the arc $xy$, $T$ contains $xy$ as well and neither $x$ nor $y$ is a leaf of $T$. Let $T'$ be the tree obtained from $T$ by contracting $xy$. 
$T'$ is an $r$-out-branching of $D'$ with at least $k$ leaves.

For the converse, let $T'$ be an $r$-out-branching of $D'$ with at least $k$ leaves. Let $x'$ be the vertex in $D'$ obtained by contracting the arc $xy$, and let $u$ be the parent of $x'$ in $T'$. 
Notice that the arc $ux'$ in $T'$ was initially the arc $ux$ before the contraction of $xy$, 
since there is no path from $r$ to $y$ avoiding $x$ in $D$. We obtain an $r$-out-branching $T$ of $D$ from $T'$,
by replacing the vertex $x'$ by the vertices $x$ and $y$ and adding the arcs $ux$, $xy$ and arc sets
$\{yz : x'z \in A(T') \wedge yz \in A(D)\}$ and
$\{xz : x'z \in A(T') \wedge yz \notin A(D)\}$.
All these arcs belong to $A(D)$ because all the out-neighbors of $x'$ in $D'$
are out-neighbors either of $x$ or of $y$ in $D$.
Finally, $x'$ must be an internal vertex of $T'$ since $x'$ disconnects $w$ from $r$.
Hence $T$ has at least as many leaves as $T'$.
\end{proof}
\subsection{Proof of Claim~\ref{claim:allarc}}
\begin{proof}
Let $T^*$ be an $r$-out-branching  of $D'$ with at least $k$ leaves and containing the maximum number of arcs from the set 
$A_c$. If $A_c\subseteq A(T^*)$, then we are through. So let us assume that there is an arc $e=xy\in A_c$ such that 
$e\notin A(T^*)$. Notice that since the vertices of ${\cal S}_{\geq 2}$ are cut vertices, they are always internal vertices in any out-branching rooted at $r$ in $D$. In particular, the vertices of ${\cal S}_{\geq 2}$ are internal vertices in $T^*$.  
Furthermore by Claim~\ref{claim:intersect} we know that $y$ is an end-point of exactly one arc in $A_c$.  Let $z$ be the 
parent of $y$ in $T^*$. Now obtain $T^*_{e}=T^*\setminus \{zy\} \cup \{xy\} $.  Observe that $T^*_{e}$ contains at least 
$k$ leaves and has more arcs from $A_c$ than $T^*$. This is contrary to our assumption that $T^*$ is an $r$-out-branching  of $D'$ with at least $k$ leaves and containing the maximum number of arcs from the set $A_c$. This proves 
that $D'$ has an $r$-out-branching  $T$ with at least $k$ leaves and containing all the arcs of $A_c$. 

Observe that starting from any $r$-out-branching  $T'$ of $D'$ we can obtain the desired $T$ in polynomial time 
by simple arc exchange operations described in the previous paragraph.  
\end{proof}
\subsection{Proof of Theorem~\ref{thm:k-path}}
\begin{proof}
Let $P$ be a path of length $k$ in a digraph $D$.
The vertex set of $P$ can be covered by at most $b$ balls of radius  $k/b$
in the metric of $UG(D)$. Let $F$ be a subgraph of  $UG(D)$  induced by the vertices contained in $b$ balls of radius  $k/b$. We claim that there is a constant $c$ (depending only on the size of the apex graph $H$), such that 
$\tw(F) \leq c \cdot k/\sqrt{b}$. Indeed, because $F$  is apex minor-free, it contains a partially triangulated 
$(d\cdot \tw(F)\times d\cdot\tw(F))$-grid as a contraction for some $d>0$ \cite{F.V.Fomin:2009eu}.   One needs $\Omega ((\tw(F) b/k)^2 )$ balls of radius  $k/b$ to cover   such a grid, and hence to cover $F$ \cite{DemaineFHT05talg}. But on the other hand, $F$ is   covered by at most $b$ balls of radius  $k/b$, and the claim follows. 
By an easy adaptation of the algorithm from 
 \cite{DornFT08} for undirected $H$-minor-free graphs, it is possible to find in time 
 $2^{\cO(\tw(F)} \cdot n^{\cO(1)}$, if the subdigraph of $D$ with the underlying undirected graph $F$ contains 
a directed path of length $k$. Thus these computations can be done in time  $2^{c_H \cdot k/\sqrt{b}} \cdot n^{\cO(1)}$ for some constant $c_H>0$ depending only on the size of $H$. 

Putting things together,   to check if $D$ contains a path of length $k$ (and if yes, to construct such a path), we try all possible sets of $b$  vertices $B$ and for each such set we construct a graph $F$ induced by vertices at distance at most $k/b$ from vertices of $B$. If $D$ contain a $k$-path, then this path should be covered by at least one such set of $b$ balls. For each such graph, we check, if the corresponding directed subgraph contains a $k$-path.  The total running time of the algorithm is 
\[
\cO(\binom{n}{b} 2^{c\cdot k/\sqrt{b}} \cdot n^{c})=
\cO(\  2^{c\cdot k/\sqrt{b}} \cdot n^{b+c})
\]    
for some constant $c$. 
By putting $b=(c/(\log(1+\ve))^2$ and $\delta=b +c$, we complete the proof of the theorem.
\end{proof}

\end{document}